\documentclass[12pt]{article}
\usepackage[arrow, matrix, curve]{xy}
\usepackage[utf8]{inputenc}
\usepackage{amssymb,amsmath,amsthm,bm,dsfont}
\usepackage{geometry}
\usepackage{enumerate}
\usepackage{mathrsfs}
\usepackage{mathtools}
\usepackage[english]{babel}
\usepackage{comment}
\usepackage[nottoc]{tocbibind}
\usepackage{hyperref}
\usepackage{afterpage}
\usepackage{fancyhdr}
\usepackage{fancyref}
\usepackage{emptypage}
\usepackage{mathtools}
\usepackage{xcolor, color}
\usepackage{tikz, pgf, pgfplots}
\usepackage{appendix}
\usepackage{chngcntr}
\usepackage{etoolbox}

\pgfplotsset{compat=1.18}

\allowdisplaybreaks[1]

\usepackage[round]{natbib}
\renewcommand{\epsilon}{\varepsilon}
\newcommand{\R}{\mathbb{R}}

\newcommand{\norm}[1]{\left\lVert #1 \right\rVert}
\newcommand{\betrag}[1]{\left| #1 \right|}

\newcommand{\E}{\mathbb{E}}

\newcommand{\eps}[2]{\epsilon( #1 \, | \, #2)}

\newcommand{\w}{\mathrm{w}}
\renewcommand{\v}{\mathrm{v}}
\newcommand{\W}{\mathrm{W}}

\renewcommand{\P}{\mathcal{P}}
\newcommand{\G}{\mathcal{G}}
\newcommand{\pup}{\overline{p}}
\newcommand{\pdown}{\underline{p}}
\newcommand{\epsp}[3]{\epsilon^{#3}( #1 \, | \, #2)}

\renewcommand{\1}{\mathrm{1}}
\newcommand{\apool}{\alpha_{\pool}}
\newcommand{\epsilonS}{\epsilon_{\Shannon}}
\newcommand{\epsilonNT}{\epsilon_{\Noisy}}

\DeclareMathOperator{\Tr}{Tr}

\DeclareMathOperator{\supp}{supp}

\DeclareMathOperator{\interim}{int}
\DeclareMathOperator{\ex}{ex}
\DeclareMathOperator{\pool}{pool}
\DeclareMathOperator{\Shannon}{S}
\DeclareMathOperator{\Noisy}{NT}
\DeclareMathOperator*{\argmax}{arg\,max}
\DeclareMathOperator*{\argmin}{arg\,min}

\definecolor{myyellow}{rgb}{1,1,0}
\definecolor{myblue}{rgb}{0,0,1}
\definecolor{myviolet}{rgb}{1,0,1}
\definecolor{mygreen}{rgb}{0,1,0}

\newtheorem{thm}{Theorem}[]
\newtheorem{lem}[]{Lemma}
\newtheorem{proposition}[]{Proposition}

\theoremstyle{definition}\newtheorem{definition}[]{Definition}
\theoremstyle{definition}\newtheorem{example}[]{Example}
\theoremstyle{definition}\newtheorem{remark}[]{Remark}
\theoremstyle{definition}

\title{Underreaction and dynamic inconsistency in communication games under noise}
\author{Gerrit Bauch\thanks{Financial support through the German Research Foundation Grant Ri-1128-9-1, the German Academic Exchange Service, and the BGTS is gratefully acknowledged. I thank Sarah Auster, Andreas Blume, Yves Breitmoser, Simon Grant, Frank Riedel, Joel Sobel  and the participants of the SAET 2023 in Paris for their comments, remarks and suggestions.}}

\begin{document}

\maketitle

\begin{abstract}
\noindent%
Communication is rarely perfect, but rather prone to error of transmission and reception. Often the origin of these errors cannot be properly quantified and is thus imprecisely known. We analyze the impact of an ambiguous noise which may alter the received message on a communication game of common interest. The noise is ambiguous in the sense that the parameters of the error-generating process and thus the likelihood to receive a message by mistake are Knightianly unknown. Ex-ante and interim responses are characterized under maxmin preferences. While the sender can disregard ambiguity, the receiver reveals a dynamically inconsistent, but astonishing behavior under a quadratic loss. Their interim actions will be closer to the pooling action than their ex-ante ones, as if facing a higher likelihood of an occurring error. 
\end{abstract}

\section{Introduction}
\noindent%
Any kind of information transmission is subject to error. It often cannot be guaranteed that messages are received as originally sent. 
The precise nature and form of the error is hardly ever known and its structural parameters usually remain uncertain. %
In this article, we investigate a cheap talk game under common interest in which an ambiguous noise channel confounds the communication between an informed sender and an uninformed action-taking receiver. While the loss minimizing communication strategy of the sender is unchanged by uncertainty, the receiver is unsettled by the uncertainty and underreacts to the received messages in an interesting way. After observing the message, the receiver responds less sensitive to the new information than originally targeted. More precisely, the receiver's actions are shifted farther towards the pooling action, i.e., the optimal action given their prior belief, after receiving a message.
Our result suggests that ambiguity can reveal behavioral traits, providing an explanation of why individuals react less to information than they expect to ex ante. Becoming afraid of having presumed the wrong worst-case situation, they re-assess their situation at the time that new information arrives. They do so in a way that makes them less confident about the received message being the one sent. Indeed, if they believe that the received message can in principle facilitate the coordination among the two agents, the worst-case is to receive that message by error. For any such \emph{regular word}, the receiver will worry about making a mistake and wants to stick closer to their prior belief, reacting less to the message than originally planned.\\
We formalize this line of thought by combining the communication model of Voronoi languages, \cite{V}, with an ambiguous noisy channel, represented by a set of stochastic noisy channels studied in \cite{NT}. Suppose there are two agents --- a sender and a receiver. While the sender has knowledge about a certain state of the world, the receiver takes an action. The agents' face a loss that depends on both the state and the state. Studying communication between the two agents, we may assume they share a common interest. We thus follow \cite{grice1975logic} who acknowledges communication to first and foremost serve a cooperative purpose. Efficient communication among the agents is frustrated by uncertainty about the precise parameters of the noisy channel. Agents seek to minimize their worst-case expected loss à la \cite{gilboa1989maxmin}. %
The communication game is modelled as an extensive form game with imprecise probabilistic information. Ambiguity formally comes into play by Nature following an Ellsberg strategy, drawing the parameters of the noisy channel from an Ellsberg urn. Information sets describe an agent's lack of knowledge about the state and/or the noise at the time their action is required. In contrast to classical literature, agents cannot form a single belief about their exact position within the information set due to ambiguity. They rather face a whole set of beliefs at each information set resulting from prior by prior updates. Following their preferences for uncertainty, they take actions to minimize their worst-case expected loss. %
As a result from full Bayesian updating, c.f., \cite{pires2002rule}, dynamic inconsistency arises, i.e., it is no longer optimal to follow the ex-ante optimal plan at the interim stage. We show that this change of worst-case beliefs follows an interesting pattern for the response of the receiver. Decomposing the optimal actions into a convex combination of the pooling action and optimal action under no error, interim actions turn out to be generically closer to the pooling action than ex-ante ones. In terms of the literature on non-Bayesian learning, c.f., \cite{epstein2010non}, agents \emph{underreact} to received information, by assigning more weight to their prior belief.\\

This article joins the economic literature on cheap talk games with common interest and puts a special focus on noise under uncertainty. Most importantly, the studied framework leans on the following two publications. Firstly, the cheap talk model of Voronoi languages, \cite{V}, in which the sender is restricted to a finite set of messages they can send, rendering perfect separation of an infinite state space impossible. As a consequence, messages are associated to a tessellation of the (normed) state space, giving it a geometric structure interpreted as a grammar, c.f., \cite{Jaeger}, \cite{gardenfors2004conceptual}. Secondly, the article on noisy talk, \cite{NT}, precludes perfect communication by introducing a class of stochastic noise. The error confounds sent and received messages, but can surprisingly lead to welfare improvements if the preferences of the sender and the receiver are not aligned.\\%
Game theory provides a neat method to model communication and has thus sparked the interest of economists for a long time. The theory of signaling games dates back to \cite{spence1978job} that has subsequently coined the notions of separating and pooling equilibria. The seminal work of \cite{CS} drops costs of signaling, shaping the class of \emph{cheap talk games}.
Imperfect communication introduces trade-offs between efficiency and understanding. \cite{nowak1999evolution} investigate an evolutionary setting and find that noise favors communication with only finitely many messages. For a finite state space, \cite{hernandez2014nash} relate a set of receiver-preferred Nash equilibria to classical communication channels à la \cite{shannon1948mathematical}. \cite{cremer2007language} study efficient communication with broad terms under bounded rationality of the sender and if the receiver faces decoding costs that increase in the breadth of the word, i.e. the number of states covered. If this loss depends on the state and the action taken rather than just the breadth, \cite{sobel2015broad} recovers convexity of the states lumped together for each message. \cite{martel2019ratings} argue for coarse communication to arise even in the absence of conflict or bounded rationality. Structural relations between a metric message space and the tessellation of the type space is studied by \cite{bauch2021effects}. \cite{rubinstein1989electronic} shows that optimal strategies can differ significantly if the common knowledge assumption is disturbed by errors in communication. Communication can also be impaired if agents are ignorant or do not share the same vocabulary, \cite{blume2013language}.\\%
It has been long known that dynamic consistency and consequentialism lead to Bayesian updating, c.f., \cite{ghirardato2002revisiting} for a discussion. Dropping dynamic consistency thus makes room for non-Bayesian learning and thus the notions for over- and underreaction. \cite{ortoleva2012modeling} studies agents who, rather than performing a Bayesian update, switch their subjective belief, referred to as paradigm change. If paradigm changes rarely occur, agents behave as-if Bayesian, c.f., \cite{weinstein2012provisional}. A nice definition of over- or underreacting is given by \cite{epstein2010non}. The importance of multiple priors for asymmetric updating has been experimentally confirmed by \cite{de2017updating}. Subjects perform a Bayesian update when signals confirm their prior, otherwise they overreact.
\\%
Cheap talk games as the one studied in this article usually have a plethora of equilibria. For instance, any equilibrium gives rise to another one by a permutation of messages and their responses. In our case, the sender's communication strategy should be measurable to guarantee integrability. To avoid issues with Bayesian updates word that are used should be sent with positive probability. Furthermore, communication should induce a tessellation consisting of convex cells to provide a convenient mathematical structure. We base these assumption on well-known linguistic concepts, primarily the \emph{cooperative principle} and its maxims of \cite{grice1975logic}, as well as on the geometric properties of languages discussed by \cite{gardenfors2004conceptual} and \cite{Jaeger}.\\

The remainder of this article is structured as follows. In Section \ref{Section: Inertia Model}, the formal model and setup is introduced. Section \ref{Section: Equilibrium Analysis} analyzes the model as an extensive form game under uncertainty and characterizes the agents' best replies. We further prove existence of efficient ex-ante equilibria. In Section \ref{Section: Inertia inconsistency}, we compare ex-ante and interim optimal responses and state as our main result that interim optimal actions are more pooling than ex-ante ones. We illustrate this insight in an example. Section \ref{Section: Inertia Conclusion} sums up our findings. The Appendix relates the stochastic version of the error channel used in this article to the one in \cite{bauch2021effects} and contains proofs of our mathematical statements as well as calculations.

\section{Model}\label{Section: Inertia Model}
\noindent%
Let $T \subseteq \R^L$ be a state space with a prior distribution $\mu_0$ with positive density function $f_0$. We assume that $T$ is a strictly convex space endowed with a norm $\norm{.}$ and compact w.r.t.\ $\norm{.}$. The sender observes their type $t \in T$ prior to sending a message $\v \in \W$ out of a finite message set $\W$ with at least two elements. During the transmission of the message, an error may confound $\v$ to a received message $\w \in \W$. The receiver interprets any observed message $\w$ as a type $s \in T$. At the end of the game, both agents face an identical loss $\ell(\norm{t-s})$ where $\ell \colon \R_{\geq 0} \to \R$ is a strictly convex function. The specific class of error channel we consider, c.f., \cite{NT}, is of the following form. If $\v$ has been sent, the receiver observes $\w$ with probability
\begin{equation}
\epsp{\w}{\v}{p,G} = (1-p) \cdot \mathds{1}_{\v}(\w) + p \cdot G(\w), \label{eq:noise channel}
\end{equation}
where $p \in [0,1]$ is the \emph{error probability} and $G \in \Delta(\W)$ an \emph{error distribution}. That is, a coin flip according to $p$ decides whether or not the message is transmitted as originally sent. In case of an error, a word is drawn from the error distribution $G$ to determine the received message. Note that the originally sent word might be drawn from $G$. A nice property of this error channel is that it is additively separable in the sense that its associated loss is the sum of the losses if an error realizes and if the message is correctly transmitted.

Strategies are denoted by a (measurable) \emph{communication device} $\pi \colon T \to \W$ for the sender and \emph{interpretation (map)} $\alpha \colon \W \to T$ for the receiver.
If the parameters $p,G$, and thus $\epsilon = \epsilon^{p,G}$, were known, the expected payoff from communication given $(\pi,\alpha)$ would be
\begin{equation}
L_{\epsilon}(\pi,\alpha) = \E_{\mu_0} \left[ \sum_{\w} \eps{\w}{\pi(t)} \cdot \ell(\norm{t-\alpha(\w)}) \right]
.
\end{equation}
Any communication device $\pi$ defines a (measurable) tessellation of the state space $T$ via its pre-images. For each message $\w \in \W$ we call $C^{\pi}(\w) \colon = \pi^{-1}(\w)$ the \emph{cell} corresponding to $\w$. If $\pi$ is understood, we frequently drop $\pi$ from the notation and simply write $C(\w)$. The \emph{support} of $\pi$ consists of all messages $\w$ that are sent with positive probability, i.e., $\mu_0(C(\w))>0$, and is denoted by $\supp(\pi)$. If $\supp(\pi) = \W$ we say that $\pi$ has a \emph{full vocabulary}. For any measurable subset $C \subseteq T$ of positive $\mu_0$-mass we can update the prior probability $\mu_0$ conditioned on the event $t \in C$. The corresponding expectation operator is denoted by $\E_{\mu_0, C}[.]$.\footnote{In the subsequent analyses, the words not in the support of $\pi$ will correspond to empty cells. If the receiver received such a word, they stick to their prior belief $\mu_0$ and expressions of the form $\mu_0(C) \cdot E_{\mu_0, C}[.]$ are set to $0$ by convention.} Finally, we let $\alpha_C \colon = \min_{s \in T} \E_{\mu_0,C}[\ell(\norm{t-s})]$ be the (unique) \emph{Bayesian estimator} of a cell $C$. The Bayesian estimator of the state space $\apool \colon = \alpha_T$ is called the \emph{pooling action}. Since the pooling action is the best response of a receiver who is ignorant of any further information except for their prior belief, it will play a crucial role as a measure for the amount of information revealed through the sender's message. The \emph{center of gravity} $\E_{\mu_0,C}[t]$ will be of particular interest later as it is the Bayesian estimator under quadratic loss.\\

We introduce uncertainty in the noisy channel by doing so for the respective parameters. Let $\mathcal{P} = [\pdown, \pup] \subseteq (0,1)$ be a set of possible error probabilities $p$ that the agents agents deem possible. The values $p=0,1$ are hereby excluded for most of the exposition as they refer to degenerate cases, corresponding either to the setting without uncertainty, already studied by \cite{V}, or where communication does not provide any information. Likewise, $\mathcal{G} \subseteq \Delta(\W)$ is a compact subset that the agents believe to contain the true error distribution. The decision makers face uncertainty about the true parameters of the noisy channel by independent\footnote{See \cite{muraviev2017kuhn} for a precise definition of independence in imprecise probabilistic settings.} draws from $\mathcal{P} \times \mathcal{G}$. Thus, $\left\{ \epsp{.}{.}{p,G} \right\}_{p \in \mathcal{P}, G \in \mathcal{G}}$ is the set of error channels the decision makers deem possible. We assume that the agents deal with uncertainty by only considering the worst-case scenario for each of their possible actions. Hence, the agents try to minimize their maximal expected loss à la \cite{gilboa1989maxmin}.

\section{Equilibrium Analysis}\label{Section: Equilibrium Analysis}
\noindent%
Our setting of communication with an ambiguous noise channel is formally described by the extensive form game illustrated in Figure \ref{Figure: extensive form game}. It contains the timing of the game and explains its information structure as follows.
\begin{enumerate}[1.]
    \item Following \cite{harsanyi1967games}, Nature moves first. A (pure) action of Nature consists of a vector $(t,c,\hat{\w})$, where $t \in T$ is the state of the world drawn from $\mu_0$, $c \in \{ 0, 1\}$ denotes if an error is realized ($c=1$) or not ($c=0$) by a draw according to $p$ and $\hat{\w} \in \W$ will be the received message in case of a realized error (i.e., if $c=1$) drawn from $G$. Knightian uncertainty comes into play by letting Nature follow an \emph{Ellsberg strategy} as defined in \cite{riedel2014ellsberg}: Nature draws their (mixed) strategy from an Ellsberg urn, whose elements are probability distributions of the form $(\mu_0, p, G) \in \{ \mu_0 \} \times \mathcal{P} \times \mathcal{G}$ over the outcomes $(t,c,\hat{\w})$.
    \item In the next step, the sender perfectly observes $t$ while facing an information set indexed by $(c,\hat{\w})$. As in \cite{muraviev2017kuhn}, we assume that the sender has a set of beliefs for each of their information sets. They can peg their action on $t$ and send a message $\pi(t) =  \v$.
    \item Afterwards and depending on $(c,\hat{\w})$, the receiver will observe a word $\w$ which is either $\v$ (if $c = 0$) or $\hat{\w}$ (if $c=1$). Thus, the receiver faces an information set indexed by the observed message $\w$ for which they also have a set of beliefs. The receiver's action is $\alpha(\w) = s$.
    \item Eventually, the loss $\ell(\norm{t-s})$ realizes for both agents.
\end{enumerate}

\begin{figure}
    \centering
    \begin{tikzpicture}[scale=4.5]
    
    	
        \coordinate (N) at (0,0);

		\coordinate (S2) at (-1.05,0);
		\coordinate (S5) at (1.05,0);		
		        
        \foreach \i [evaluate=\i as \j using int(\i+3),evaluate=\i as \k using int(\i-2)] in {1,3}{
        	\coordinate (S\i) at (-0.8,0-\k/4);
        	\coordinate (S\j) at (0.8,0+\k/4);
        }

		\coordinate (R1) at (0,1);
		\coordinate (R2) at (-0.5,1);
		\coordinate (R6) at (0.5,1);
		\coordinate (U1) at (0,1.5);
		\coordinate (U2) at (-0.5,1.5);
		\coordinate (U6) at (0.5,1.5);

        \foreach \i in {3,4,5}{
        	\coordinate (R\i) at (-0.5+\i/2-3/2,-1);
        	\coordinate (U\i) at (-0.5+\i/2-3/2,-1.5);
        }

        \filldraw (N) circle (0.6pt);
        
        \foreach \i in {1,...,6}{
        	\filldraw (S\i) circle (0.6pt);
        	\filldraw (R\i) circle (0.6pt);
        	\filldraw (U\i) circle (0.6pt);
        }

        \foreach \i in {1,...,3}{
        	\draw (S\i)--node[midway,left]{$\pi(t)$}(R\i);
        }
        \foreach \i in {4,...,6}{
        	\draw (S\i)--node[midway,right]{$\pi(t')$}(R\i);
        }

        \foreach \i in {1,2,6}{
            \draw (R\i)--node[left,midway]{$\alpha(\w)$} (U\i);
        }

        \foreach \i in {3,4,5}{
            \draw (R\i)--node[left,midway]{$\alpha(\w')$} (U\i);
        }

        \draw (N) -- node[midway,above]{$(t,0,\w)$}(S1);
        \draw (N) -- node[midway,above]{$(t,1,\w)$} (S2);
        \draw (N) -- node[midway,below]{$(t,1,\w')$} (S3);
        \draw (N) -- node[midway,below]{$(t',0,\w')$} (S4);
        \draw (N) -- node[midway,above]{$(t',1,\w')$} (S5);
        \draw (N) -- node[midway,above]{$(t',1,\w)$} (S6);
        
        \foreach \i in {-1,1}{
	        \draw (\i*0.9,0) ellipse[x radius=0.35cm, y radius=0.6cm];
	    	\draw (0,\i) ellipse[x radius=0.65cm, y radius=0.2cm];
	    }

        \node at (N) [above,shift={(0,0.2)}]{Nature};
		\foreach \i in {-1,1}{
			\node at (\i*1.4,0) []{Sender};
			\node at (0,\i*0.7) []{Receiver};
        }

		\node at (U6) [above,scale=0.8]{$\ell(\norm{t'-\alpha(\w)})$};
		\node at (U3) [below,scale=0.8]{$\ell(\norm{t-\alpha(\w')})$};
		\foreach \i [evaluate=\i as \j using int(6-\i)] in {4,5}{
			\node at (U\j) [above,scale=0.8]{$\ell(\norm{t-\alpha(\w)})$};
			\node at (U\i) [below,scale=0.8]{$\ell(\norm{t'-\alpha(\w')})$};
		}

    \end{tikzpicture}
    \caption{Outline of the communication game under an uncertain noise in extensive form with information sets. Nature picks a state $t \sim \mu_0$, decides whether or not an error occurs ($c=0,1$) and chooses an error message $\hat{\w}$ ($\w$ or $\w'$). Using an Ellsberg strategy over the error, she keeps the agent uncertain about her mixed strategy $(p, G) \in \mathcal{P} \times \mathcal{G}$ over $(c,\hat{\w})$ which is drawn from an Ellsberg urn and leads to sets of beliefs at each information set. The sender chooses a word to be sent according to $\pi$ and the perfectly observed state $t$ while the receiver picks an interpretation $\alpha$ according to the word received.}
    \label{Figure: extensive form game}
\end{figure}
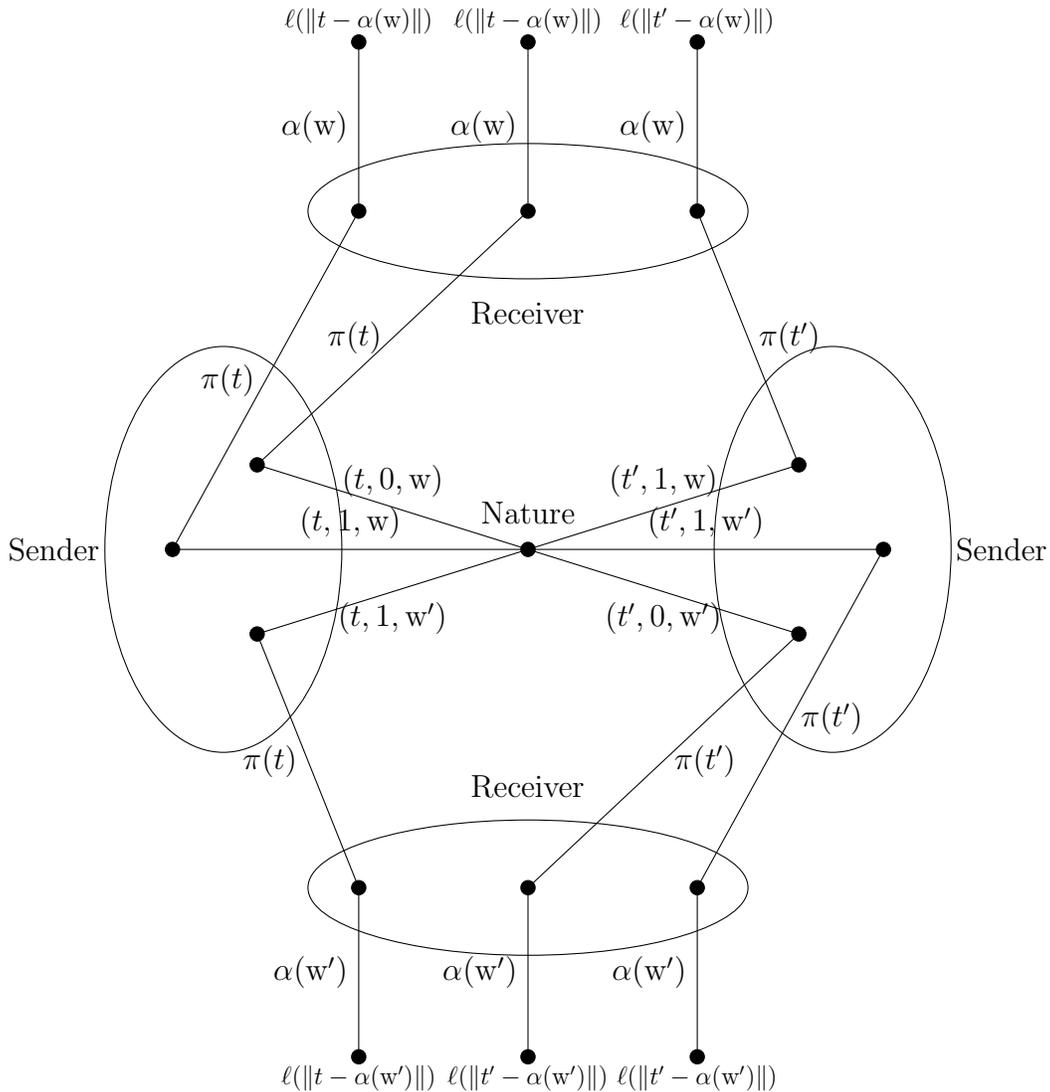

In equilibrium, we require two properties. Firstly, the agents' sets of beliefs at each of their information sets need to be derived from prior moves, i.e., from Nature's Ellsberg strategy and, for the receiver, sender's strategy. Secondly, the strategy of each player must be loss minimizing at any of their information sets. We discuss the belief updates the agents face at each of their information sets and the corresponding optimal plays in the upcoming subsections.

\subsection{Information Sets and Belief Updates}
\noindent%
Given the Ellsberg strategy of Nature and following full Bayesian updating, the sender who observes the true state $t \in T$ faces the set $\mathcal{P} \times \mathcal{G}$ of distributions over their position within the information set given by $\{ (t,c,\hat{\w}) \mid c \in \{0,1\}, \hat{\w} \in \W \}$, where $c$ is drawn from a coin toss with probability $p \in \mathcal{P}$ and the error message $\hat{\w}$ is drawn from $G \in \mathcal{G}$.\\
Likewise, the receiver, who observes some $\w \in \W$, cannot distinguish whether or not $\w$ was the message the sender intended to convey. Thus the receiver faces the set of probability measures $\{\mu_0\} \times \mathcal{P} \times \mathcal{G}$ conditioned on having observed $\w$ and knowing $\pi$ in equilibrium. Note that the receiver's payoff-maximizing action will only depend on the induced marginal beliefs on $T$. Each fixed $(p,G) \in \mathcal{P} \times \mathcal{G}$ yields one of these posterior marginal belief(s), denoted by $\mu^{\pi,p,G}_{\w}$ and described via its density
\begin{align}
    f^{\pi,p,G}_{\w}(t)\colon = 
     \frac{\epsp{\w}{\pi(t)}{p,G}}{\int_T \epsp{\w}{\pi(t)}{p,G} \, \mu_0(\mathrm{d}t)} \cdot f_0(t)
    = \frac{(1-p) \cdot \1_{\w}(\pi(t)) + p \cdot G(\w)}{(1-p) \cdot \mu_0(C(\w)) + p \cdot G(\w)}\cdot f_0(t),\label{eq: updated posterior}
\end{align}
if $(1-p)\cdot \mu_0(C(\w)) + p G(\w)$ --- the expected frequency of observing $\w$ under $(\pi, p,G)$ --- is non-zero. This is especially the case for every $\w \in \supp(\pi)$. In equilibrium and for an observed word $\w$, the receiver's best reply depends only on the set of posterior beliefs over $T$, which is given by $\{ \mu^{\pi,p,G}_{\w} \}_{p,G}$.\bigskip

\subsection{Best Replies}
\noindent%
Besides beliefs being in line with the play, we require the agents' behavior to be rational, i.e., playing a best response to their peer's strategy in order to minimize their loss in communication. Under Knightian uncertainty a readily available decision rule is to pick any action minimizing the worst-case (here: the maximal) expected loss, c.f., \cite{gilboa1989maxmin}. In contrast to usual Bayesian extensive form games with incomplete information and fully supported marginals, an initially optimal strategy might no longer stay so when an agent received their information in a setting of multiple priors, c.f., \cite{pahlke2022dynamic}.\\

To understand the arise of this phenomenon better, one distinguishes two instances of time --- before or after the agent has received information and is required to act. At the ex-ante stage, an agent has not yet observed their private information. Any strategy that they might think of before their play is a contingent plan: They assign an action to each possible future state realization. Being rational, they seek to minimize their worst-case expected loss given their system of beliefs from the very beginning of the game. At the interim stage, an agent has observed their private information and seeks to play an action that minimizes their expected loss given the information at hand. As we will see, these approaches do not need to amount to the same behavior.

\paragraph{Ex-ante best reply}
\noindent%
At the ex-ante stage, players contemplate their best contingent plan to minimize their maximal worst-case loss at (or even before) the very beginning of the game. Facing the peer's equilibrium strategy $\alpha$ resp.\ $\pi$, the agents' ex-ante worst-case expected losses are equal and given by
\begin{align}
    &U_{\ex}^S (\pi,\alpha) = U_{\ex}^R (\pi,\alpha)\nonumber\\
    =&\max_{(p,G) \in \mathcal{P} \times \mathcal{G}} (1-p) \cdot \E_{\mu_0}[\ell(\norm{t-\alpha(\pi(t))})] + p \cdot \sum_{\w} G(\w) \cdot \E_{\mu_0}[\ell(\norm{t-\alpha(\w)})],\label{eq: ex ante worst-case utility}
\end{align}

We examine the above expression from each agent's point of view, starting with the sender. From \eqref{eq: ex ante worst-case utility}, it becomes clear that the sender's strategy $\pi \colon T \to \W$ only impacts their loss if the message is correctly transmitted. Consequently, a best reply for the sender minimizes
\begin{align*}
    \E_{\mu_0}[\ell(\norm{t - \alpha(\pi(t))})]
\end{align*}
irrespective of $p,G$. The sender can achieve this point-wise in the integrand by picking
\begin{equation}
    \pi(t) = \v \in \argmin_{\w} \norm{t-\alpha(\w)} \label{eq: ex-ante sender optimization}
\end{equation}
for every $t \in T$, indeed, the sender needs to do this $\mu_0$-almost surely. It is worth noting that this requirement is independent of uncertainty. In fact, it is exactly the behavior of a sender in the classical case of Voronoi languages without a noisy communication channel, c.f., \cite{V}.\\

Among the best replies to $\alpha$, we can always choose a communication device $\pi$ fulfilling the following properties that are not only mathematically convenient, but also formalize aspects of linguistic theories. Firstly, $\pi$ is measurable, c.f., \cite{V} or \cite{bauch2021effects}, thus avoiding \emph{obscurity of expression}, a submaxim of the \emph{cooperative principle}, c.f., \cite{grice1975logic}. Secondly, cells with measure zero are empty, following the \emph{maxim of quantity} of the cooperative principle. Thirdly, for certain types of norms, such as the Euclidean one, the induced cells are convex with zero-measure boundaries, formalizing the idea of \emph{convex categories}, c.f., \cite{gardenfors2004conceptual} and \cite{Jaeger}. In the following, we are going to assume that $\pi$ is always chosen to abide by the above properties.\\

In contrast, the presence of uncertainty influences the receiver's action choices significantly. For equilibrium considerations, we take into account the optimal strategy of the sender. The worst-case value of the error probability $p$ can then readily be identified as the maximal value $p = \pup$, yielding the following ex-ante worst-case expected loss for the receiver
\begin{align}
    U_{\ex}^R(\pi, \alpha) = (1-\pup) \cdot \E_{\mu_0}[\min_{\v} \ell(\norm{t-\alpha(\v)})] + \pup \cdot \max_{G} \sum_{\w} G(\w) \cdot \E_{\mu_0}[\ell(\norm{t-\alpha(\w)})].\label{eq: ex ante worst-case receiver simplified}
\end{align}
It is sometimes convenient to rewrite $U_{\ex}^R(\pi, \alpha)$ in the following form
\begin{align}
    \max_{G} \sum_{\w} (1-\pup) \mu_0(C(\w))\cdot \E_{C(\w)}[\ell(\norm{t-\alpha(\w)})] + \pup G(\w) \cdot \E_{\mu_0}[\ell(\norm{t-\alpha(\w)})],\label{eq: ex ante worst-case receiver similar to interim}
\end{align}
where the sender's strategy has been re-written by means of the according cells.
\\

The involved worst-case distributions are given as follows. Firstly, under equilibrium considerations the worst-case invokes the maximal error probability $p$, i.e., making the occurrence of an error most likely since that restricts the communication of the sender the most. Secondly, a worst-case error distribution induces a painful trade-off for the receiver between following the designated action $\alpha(\w)$ and sticking to the prior belief $\mu_0$. It chooses a distribution of weights that maximizes the punishment from deviating from $\mu_0$, making the receiver be in favor of actions close to the pooling one while being less separating. For instance, if $\mathcal{G}$ contains all distributions, then a worst-case distribution may assign all weight to that word $\w$, leading to the action that is worst in terms of the prior belief. In the case of a quadratic loss, those actions receive the most weight that induce actions furthest away from the pooling action.\\

We are going to establish existence and uniqueness of the best reply of a receiver to the ex-ante problem together with the interim one in Proposition \ref{Proposition: Existence and uniqueness of best replies}.

\paragraph{Interim best reply}
\noindent%
At the interim stage, the agents have received their private information. The sender will know the state $t$ while the receiver has observed a message $\w$. It is helpful to think of the interim stage as one with more players --- one player for each piece of private information. Everyone of them will play on their own to minimize their loss. In our case agents act on basis of their information by sending a word $\v$ resp.\ picking an action $s$ that minimizes their worst-case expected loss given the information at hand. Explicitly, the interim worst-case expected loss for the sender who observed the state $t$ and facing receivers who will respond with a vector $\{\alpha(\w)\}_{\w}$ of actions is given by
\begin{align}
    U_{\interim, t }^S(\v, \alpha) = \max_{p,G} \ (1-p) \cdot \ell(\norm{t - \alpha(\v)}) + p \cdot \sum_{\w} G(\w) \cdot \ell(\norm{t-\alpha(\w)}).\label{eq: interim worst-case sender}
\end{align}
The expression is minimized as in the ex-ante case by picking 
\begin{equation}
    \v \in \argmin_{\w} \norm{t-\alpha(\w)}. \label{eq: interim sender optimization}
\end{equation}

The receiver who observes the message $\w \in \W$ updates their belief about the state of the world. As this update depends on the configuration of the unknown error channel, we assume that the agent performs a full Bayesian update, i.e., they update their belief for every possible error profile separately and thus face a set of posterior beliefs $\{ \mu^{\pi,p,G}_{\w} \}_{p,G}$ over the state space $T$. Formally, the posterior belief after observing $\w \in \supp(\pi)$ is given by its density function $f^{\pi,p,G}_{\w}$, c.f., expression \eqref{eq: updated posterior}. The interim worst-case expected loss of the receiver after observing $\w \in \supp(\pi)$ can be written as a convex combination of the payoff given by the cell $C(\w)$ and the whole state space $T$:
\begin{align}
    U_{\interim,\w}^R (\pi, s) &= \max_{p,G} \ \E_{\mu^{\pi,g,G}_{\w}}[\ell(\norm{t-s})] \nonumber\\
    &= \max_{\kappa \in [\underline{\kappa},\overline{\kappa}]} \ (1-\kappa) \cdot \E_{\mu_0, C(\w)}[\ell(\norm{t-s})] + \kappa \cdot \E_{\mu_0}[\ell(\norm{t-s})]\label{eq: interim worst-case receiver simplified}
\end{align}
with $\kappa\colon =\kappa^{\pi, p,G}_{\w} \colon = \frac{p \cdot G(\w)}{(1-p) \mu_0(C(\w)) + pG(\w)}$. Note that $\kappa$ is increasing in $p$ and $G(\w)$, thus $\kappa$ runs in the compact interval from $\underline{\kappa}$ for $p=\pdown$ and $\underline{G}(\w) \colon = \min_G G(\w)$ to $\overline{\kappa}$ for $p=\pup$ and $\overline{G}(\w) \colon = \max_G G(\w)$.\bigskip

Since the parameter $\kappa$ controls the weights of the convex combination in the interim problem, one can restrict the worst-case set to $\{ \underline{\kappa}, \overline{\kappa} \}$. Despite this simplification, which value $\kappa$ attains in the worst-case is still a delicate matter. There are three possible options, depending on whether $\E_{\mu_0}[\ell(\norm{t-s})]$ is larger than, equal to or smaller than $\E_{\mu_0,C(\w)}[\ell(\norm{t-s})]$. Indeed, which action ultimately will be chosen by an ambiguity averse receiver has to do with the structure of the cells themselves, as we will analyze later in Section \ref{Section: Inertia inconsistency}.

\paragraph{Existence and Uniqueness of Best Replies}
\noindent%
In equilibrium, the sender's optimal strategy is to always send any word leading to the action closest to their type. This strategy is independent of uncertainty and $\mu_0$-almost surely unique if the receiver's actions are pairwise distinct and the Euclidean is used, c.f., \cite{V}. At any rate, ex-ante and interim optimal sender communication agree.\\
The ex-ante and interim problem of the receiver, given by \eqref{eq: ex ante worst-case receiver simplified} and \eqref{eq: interim worst-case receiver simplified}, also admit unique solutions, as summarized by the following Proposition and its discussion.

\begin{proposition}\label{Proposition: Existence and uniqueness of best replies}
    The receiver's ex-ante and interim optimal replies exist. The resp.\ best reply to $\w$ is unique if $\mu_0(C(\w)) >0$. If $\mu_0(C(\w)) = 0$, $\apool$ is a best reply.

\begin{proof}
        All proofs and calculations are delegated to the appendix.
    \end{proof}
\end{proposition}

Recall that we have restricted our attention to communication devices that entail empty zero-measure cells, following the maxim of quantity. Any such message $\w$ that is never sent (i.e., $C(\w) = \emptyset$) should be replied to with $\apool$: On the one hand side, if $\w$ is expected to be received by (a worst-case) error, $\apool$ is the unique best reply. If, on the other hand side, $\w$ is not expected to be received under (a worst-case) error either, the receiver's response does not alter the (worst-case) expected payoff and is thus free (in an environment of $\apool$). However, if the receiver does actually face $\w$, 
they do not know how to re-assess the situation and thus stick to the prior belief. In that sense, replying with the pooling action $\apool$ is \emph{sequentially rational}. Proposition \ref{Proposition: Existence and uniqueness of best replies} hence suggests the following notion of unique best replies which we will study throughout the rest of our analysis.

\begin{definition}\label{Definition: ex ante and worst-case reply uniqueness}
    We let $\alpha_{ex}$ resp.\ $\alpha_{\interim}$ denote the unique best ex-ante resp.\ interim reply to $\pi$ that assigns $\alpha_{ex}(\w) = \alpha_{\interim}(\w) = \apool$ whenever $\mu_0(C(\w)) = 0$.
\end{definition}

\subsection{Existence of non-babbling Equilibria}
\noindent%
Having characterized the best replies of the agents from an ex-ante and interim perspective, we employ the usual concept of equilibrium by means of mutual best replies.

\begin{definition}
    A strategy profile $(\pi,\alpha)$ is an \emph{ex-ante (worst-case) equilibrium}, if it minimizes both ex-ante worst-case expected losses, i.e.,
    \begin{align*}
        U_{\ex}^S (\pi,\alpha) &\leq U_{\ex}^S (\pi',\alpha), \qquad \text{for all } \pi',\\
        U_{\ex}^R (\pi,\alpha) &\leq U_{\ex}^R (\pi,\alpha'), \qquad \text{for all } \alpha'.
    \end{align*}
    A strategy profile $(\pi,\alpha)$ is called \emph{ex-ante (Pareto-)efficient}, if there is no $(\pi',\alpha')$ with higher ex-ante worst-case expected loss for both agents, at least one being strict.
\end{definition}

\begin{definition}
    A strategy profile $(\pi,\alpha)$ is an \emph{interim (worst-case) equilibrium}, if it minimizes both interim worst-case expected losses \eqref{eq: interim worst-case sender} and \eqref{eq: interim worst-case receiver simplified}, i.e.,
    \begin{align*}
        U_{\interim, t}^S (\pi(t),\alpha) &\leq U_{\interim, t}^S (\v,\alpha), \qquad \text{for all } t \in T, \v \in \W,\\
        U_{\interim, \w}^R (\pi,\alpha(\w)) &\leq U_{\interim, \w}^R (\pi,s), \qquad \text{for all } \w \in \W, s \in T.
    \end{align*}
\end{definition}

Note that in cheap talk games, equilibria without proper communication always exist. These \emph{babbling equilibria} are characterized by a sender who will send their messages independently of their type (i.e., $\pi(t)$ is a fixed (mixed) strategy) and receiver responds optimally by playing the pooling action independent of the observed message. In all these equilibria, the expected loss will be $L_0 = \E_{\mu_0}[\ell(\norm{t-\apool})]$, irrespective of the error channel. In order for communication under an uncertain error channel to be \emph{meaningful}, we aspire at the existence of non-babbling equilibria. The following proposition asserts that the Pareto-efficient equilibrium exists and is non-babbling.

\begin{proposition}\label{Proposition: Existence of efficient ex-ante wc equilibrium}
    A Pareto-efficient ex-ante worst-case equilibrium exists and is non-babbling.
\end{proposition}

Unfortunately, we were not able to find an interesting notion of efficiency for the interim stage for which we could give an existence proof. Part of the reason for that is that, although we generally face a setting of common interest, the receiver's incentives at the interim stage differ from their ex-ante ones. All we can say for certain is that interim equilibria exists because of babbling equilibria.

\section{Dynamic Inconsistency}\label{Section: Inertia inconsistency}
\noindent%
When analyzing extensive form games under Knightian uncertainty, actions that were deemed optimal before one's turn may not stay so after obtaining one's information. Any such deviation from ones ex-ante optimal plan is called \emph{dynamic inconsistency}, c.f., \cite{pahlke2022dynamic}, \cite{epstein2003recursive}, \cite{aryal2014note}. As already established, the sender does not behave in a dynamically consistent way.\\
In the following, we investigate whether the receiver exerts dynamically inconsistent behavior under an uncertain noise. Indeed, we find that the receiver's ex-ante optimal plan does generically not stay optimal after receiving a message. Even more, interim actions will be more pooling than the ex-ante optimal play, revealing that the received information will be less convincing in moving away from the pooling action. To this end, fix a communication device $\pi$ throughout this whole section.\\

Recall that the agents know that the true noise channel follows the structure of expression \eqref{eq:noise channel}, but are uncertain about the exact values of $p$ and $G$. They only believe that $(p,G)$ lies in a set $\P \times \G$ without any further knowledge about the parameters. Both agents are ambiguity averse in the sense of \cite{gilboa1989maxmin}, thus they try to minimize their worst expected loss.

\subsection{Quadratic Loss}
\noindent%
Let's turn again to the case $\mathcal{P} = [\pdown, \pup] \subseteq (0,1)$. We furthermore assume a quadratic loss, i.e., $\ell \circ \norm{.} = \norm{.}_2^2$. In that case, we have that for any cell $C$ of positive $\mu_0$-mass
\begin{equation}
    \E_{\mu_0, C}[\ell(\norm{t-s})] = \E_{\mu_0, C}[\norm{t-s}_2] = \E_{\mu_0,C}[\norm{t}_2^2] - \norm{\E_{\mu_0,C}[t]}_2^2 + \norm{\E_{\mu_0,C}[t] - s}_2^2.
\end{equation}
Consequently, the Bayesian estimator is the center of gravity $\alpha_C = \E_{\mu_0, C}[t]$ and the above function is determined solely by the distance of $s$ to $\alpha_C$. We are going to identify a loss in confidence revealed through dynamic inconsistency in sender-receiver games under Knightian uncertainty about the error channel. As it turns out, interim actions are going to be closer to the pooling action than ex-ante ones.\\

Our first observation is that the ex-ante and interim optimal actions $\alpha_{\ex}(\w)$ and $\alpha_{\interim}(\w)$ always lie on the line segment between $\apool$ and the center of gravity $\alpha_{C(\w)}$. Not only does this reduce the complexity of finding the best receiver actions to one dimensional problems for each message $\w \in \W$, but also provides a measure for 'how pooling' an action is.

\begin{proposition}\label{Proposition: Optimizer on line segment generalized in proof}
    The best responses $\alpha_{\ex}(\w)$ and $\alpha_{\interim}(\w)$ lie on the line segment between $\apool$ and $\alpha_{C(\w)}$ for all $\w$.
\end{proposition}

Following the \emph{maxim of quantity ('Be informative.')} and the \emph{maxim of relation ('Be relevant.')} of \cite{grice1975logic}, agents aspire to communicate in a way that maximizes information and relevance of the words in place. For example, in the extreme case that a communication device induces only one action, the sender does neither provide any information nor is it relevant (as only the pooling action is induced). It thus stand to reason, that they try to maximize the information transmitted by means of the spread of the actions induced. Indeed, efficient communication tries to be as 'separating as possible', c.f., \cite{V}. At the same time, however, the noisy channel confounds this effort, necessitating to prevent losses due to misunderstandings. Proposition \ref{Proposition: Optimizer on line segment generalized in proof} provides a way of quantifying this trade-off by means of the distance of $\alpha_{\ex}(\w)$ resp.\ $\alpha_{\interim}(\w)$ to the pooling action $\apool$ for a received word $\w$: If the risk of misunderstandings is moderately low, 
the best (ex-ante or interim) reply is going to be close to $\alpha_{C(\w)}$ --- the action chosen if there were no noise. If errors in communication are likely, the receiver is less confident that the received message was originally sent and reduces potential losses due to misunderstandings by choosing an action close to $\apool$ --- the action chosen if there were no communication.\\

Another notion facilitated by Proposition \ref{Proposition: Optimizer on line segment generalized in proof} is that of \emph{underreaction}, c.f., \cite{epstein2010non}. The weight put on the pooling action or the center of the cells can be interpreted as a trade-off between keeping the prior belief $\mu_0$ and performing the Bayesian update to $\mu_{0,C(\w))}$ in absence of noise. The less weight is put on the updated belief, the less an individual reacts to the arrival of new information.

\subsection{Regular Words}
\noindent%
The main result of our analysis states that the interim optimal actions to a communication device will be closer to the pooling action than the ex-ante ones. In light of the above discussion, this can be interpreted as the receiver, at the time of receiving a message, losing more confidence in the signal than originally anticipated.

This result is driven by the following assumption on a word $\w$ and its cell that helps pin down the worst-case belief at the interim stage. We assume that the interim worst-case will be $\overline{\kappa}$, i.e., putting maximal weight on the prior probability and minimal weight on the conditional probability of the cell itself (without any error). This means that the expected loss conditioned on the cell $C(\w)$ itself is lower than the loss for the whole state space. As it seems natural to assume that the worst-case of a noisy channel is to make it most probable that the received message arrived by mistake, we call such a word \emph{regular}.

\begin{definition}\label{Definition: regular word}
    A message $\w \in \W$ is called \emph{regular} (under $\pi$) ifwe have  $\alpha_{\interim}(\w) = (1-\overline{\kappa}) \alpha_{C(\w)} + \overline{\kappa} \apool$, i.e., if the interim worst case invokes $\overline{\kappa} = \tfrac{\pup \overline{G}(\w)}{(1-\pup) \mu_0(C(\w)) + \pup \overline{G}(\w)}$.
\end{definition}

The following lemma gives a sufficient technical condition for a word to be regular and establishes their existence.
\begin{lem}\label{Lemma: regular words properties}
A word $\w$ is regular under $\pi$, if
    \begin{equation}
        (2\overline{\kappa}-1) \cdot \norm{\alpha_{C(\w)}-\apool}_2^2 \leq \Tr_{\mu_0} - \Tr_{\mu_0,C(\w)},\label{eq: condition regular word}
    \end{equation}
    where we denote by $\Tr_{\mu} \colon = \E_{\mu} \left[ \norm{t - \E_{\mu}[t]}_2^2 \right]$ the trace norm of the variance matrix for any probability measure $\mu$ on $T$. Furthermore, a regular word always exists.
\end{lem}

The trace norm of the variance matrix is the sum of the variances of all one-dimensional projections and provides a scalar-valued measure of the total variance. Lemma \ref{Lemma: regular words properties} states that a word is regular, whenever the total variance of its cell is small enough and, in addition if $\overline{\kappa} > \tfrac{1}{2}$, the corresponding center of gravity is not too far away from the pooling action.

\subsection{Interim actions are more pooling than ex-ante ones}
\noindent%
It is within the nature of extensive form games to distinguish between the ex-ante and the interim stage for an agent's decision making, \cite{fudenberg1991game}. Having studied the nature of the ex-ante and interim replies, we now want to compare the two to one another. In contrast to a classical Bayesian setting with fully supported marginals, ex-ante and interim optimal actions do not necessarily coincide in a setting of Knightian uncertainty with full Bayesian updates. Indeed, we are now going to unravel generic dynamically inconsistent behavior in our cooperative communication game. Our main theorem identifies interim actions to be closer to the pooling action, and thus less separating than the ex-ante ones.\\

In order to provide an exhaustive analysis, we first discuss the particular limit cases when there is no error ($p=0$) or noise only ($p=1$).\\
If $p=0$, there is no error and uncertainty about the error does not have a bite in the communication game. The receiver plays the Bayesian estimators of each cell in response to a received message, as in the benchmark case of Voronoi languages, \cite{V}. Facing fully supported marginals, ex-ante and interim optimal action coincide and no dynamic inconsistency arises.\\
If $p=1$, the error channel $\epsp{.}{\v}{p,G}$ is independent of the sent word $\v$. The sender is thus indifferent between any communication device as they cannot influence the transmission in any way. In turn, the receiver cannot infer anything out of any received signal. Sticking to their prior belief and taking the pooling action is thus optimal. Consequently, ex-ante and pooling action of the receiver coincide and are again equal for every word and in this case always coincide with the pooling action.\\

We now return to our original setting where $[\pdown,\pup] \subseteq (0,1)$ and state the main result of this article. We show that the receiver behaves dynamically inconsistent if facing an uncertain noise. The inconsistency in their behavior exhibits a remarkable pattern. After receiving information, the receiver changes their worst-case belief to one that puts more weight on their prior belief. As a consequence, the receiver values the new information less and underreacts to the signal. Their optimal interim action will be strictly closer to the pooling action in comparison to their ex-ante optimal plan. Our interpretation of this is that at an interim stage information is considered less valuable or convincing then from an ex-ante point of view.

\begin{thm}\label{Theorem: interim strictly more pooling than ex ante}
    The receiver's interim best replies are more pooling than the ex-ante ones, i.e., $\norm{\alpha_{\interim}(\w) - \apool}_2 \leq \norm{\alpha_{\ex}(\w)- \apool}_2$ for all regular $\w \in \W$.\\
    If $\pi$ has a full vocabulary, all words are regular, none is the pooling action and $\#\mathcal{G} \geq 2$, then the above inequality is strict for all $\w$.
    
\end{thm}

\subsection{Example}
\noindent%
We now give an illustrative example to demonstrate that the receiver's behavior is dynamically inconsistent if they face Knightian uncertainty. Detailed calculations are given in the appendix.

\begin{example}\label{Example: Inconsistency on unit interval}
A company produces a good and seeks information about near-future consumption. Depending on the exact future demand, the optimal production amount lies in the interval $T = [-\tfrac{1}{2},\tfrac{1}{2}]$. The company has a uniform prior $\mu_0$ on $T$ and thus does not change the production if it has no information, i.e., $\apool = 0$. The company faces a quadratic loss of not producing the right amount, i.e., $\ell \circ \norm (.) = \norm{.}_2^2$. The company's customer analytics team (sender) observes the change in demand. However, communication with management (receiver) is limited to broad terms, telling them to either lower ($L$) or raise ($R$) production. Afraid that analytics' recommendation is wrong with a fixed probability $p$ ($\mathcal{P} = \{p\}$) and facing full uncertainty about the suggestion if an error happened ($\mathcal{G} = \Delta(\{L,R\})$), management faces ambiguity about their action.%

Strategizing before the arrival of information, the company finds the following (up to re-labeling of the messages almost surely unique) ex-ante Pareto-efficient equilibrium featuring the communication device
\begin{equation*}
\pi(t) = \begin{cases}
L &, t<0,\\
R &, t \geq 0
\end{cases}
\end{equation*}
and the ex-ante optimal actions $\alpha_{\ex}(R) = \tfrac{1}{4}\cdot (1-p) = -\alpha_{\ex}(L)$. Under $\pi$, condition \eqref{eq: condition regular word} is fulfilled for each word and thus they are both regular. Consequently, the interim optimal actions are given by $\alpha_{\interim}(R) = \tfrac{1}{4} \cdot \tfrac{1-p}{1+p}= - \alpha_{\interim}(L)$. Indeed, since $\alpha_{\interim}$ is symmetric around $0$, ($\pi,\alpha_{\interim})$ forms an interim equilibrium. Figure \ref{Fig: dynamic inconsistency for uncertainty in G} illustrates the equilibria for all fixed values of $p \in [0,1]$.\\
Interim optimal actions are closer to the pooling action $\apool = 0$ than ex-ante ones. For a fixed error probability $p_1$, they are equal to ex-ante optimal actions for a higher error probability $p_2$. Our interpretation is that management has more trust in the signal at the ex-ante stage than at the interim stage, ultimately resulting in a less pronounced response.\bigskip

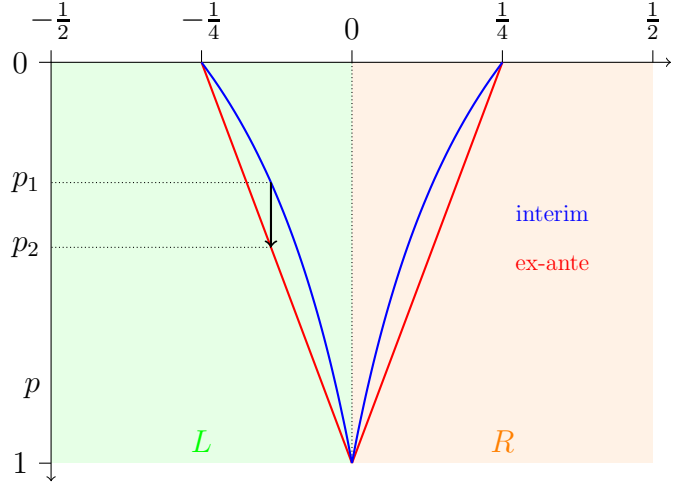
\begin{figure}
    \begin{center}
\begin{tikzpicture}[scale=8,rotate=-90]

\fill[green,opacity=0.1] (0,-1/2)--(0,0)--(2/3,0)--(2/3,-1/2)--cycle;
\fill[orange,opacity=0.1] (0,1/2)--(0,0)--(2/3,0)--(2/3,1/2)--cycle;
\draw (2/3,-1/4)node[above,green]{$L$};
\draw (2/3,1/4)node[above,orange]{$R$};

\draw[<->] (0,1/2+0.03)--(0,-1/2)--(2/3+0.03,-1/2)node[left, yshift=3em]{$p$};

\draw (0,1/4)--(-0.02,1/4)node[above]{$\frac{1}{4}$};
\draw (0,-1/4)--(-0.02,-1/4)node[above]{$-\frac{1}{4}$};
\draw (-0.02,0)node[above]{$0$}--(0,0);
\draw[densely dotted](0,0)--(2/3,0);
\draw (0,-1/2)--(0,-1/2-0.02)node[left]{$0$};
\draw (2/3,-1/2)--(2/3,-1/2-0.02)node[left]{$1$};

\draw (0,1/2)--(-0.02,1/2)node[above]{$\frac{1}{2}$};
\draw (0,-1/2)--(-0.02,-1/2)node[above]{$-\frac{1}{2}$};

\draw[scale=1, domain=0:1, smooth, variable=\x, red,thick] plot ({ 2/3 * \x}, {1/4 * (1 - \x)});
\draw[scale=1, domain=0:1, smooth, variable=\x, red, thick] plot ({ 2/3 * \x}, {-1/4 * (1 - \x)});
\draw (1/3,1/3)node[red,scale=0.75]{ex-ante};

\draw[scale=1, domain=0:1, smooth, variable=\x, blue,thick] plot ({2/3* \x}, {1/4 * (1 - \x)/(1 + \x)});
\draw[scale=1, domain=0:1, smooth, variable=\x, blue,thick] plot ({ 2/3*\x}, {-1/4 * (1 - \x)/(1 + \x)});
\draw (1/4,1/3)node[blue, scale=0.75]{interim};

\draw[densely dotted] (2/3 * 0.3, -0.1346)--(2/3 * 0.3, -1/2)node[left]{$p_1$};

\draw[thick,->](2/3 * 0.3, -0.1346)--(2/3 * 0.4616, -0.1346);

\draw[densely dotted] (2/3 * 0.4616, -0.1346)--(2/3 * 0.4616, -1/2)node[left]{$p_2$};

\end{tikzpicture}
\end{center}
    \caption{Ex-ante Pareto-efficient and interim equilibria given the communication device, splitting the state space in two (L(eft) and R(ight)) intervals. For every fixed $p \in [0,1]$, the uncertainty set is given by $\mathcal{P} = \{ p \}$ and $\mathcal{G} = \Delta(\{ L,R \})$. Except for $p=0,1$, interim optimal actions are stricly more pooling, i.e., closer to $\apool = 0$ than ex-ante ones, revealing dynamically inconsistent behavior.}
    \label{Fig: dynamic inconsistency for uncertainty in G}
\end{figure}
\end{example}

\section{Conclusion \& Discussion}\label{Section: Inertia Conclusion}
\noindent%
This article studies a communication game of common interest under ambiguous noise. We formalize the information structure as an extensive form game under Knightian uncertainty, letting the agents face sets of beliefs for each of their information sets. While the sender is unperturbed by the presence of uncertainty, the receiver's response is significantly impacted. Being afraid of having received a message by error, they will never fully commit to the message. They rather play an action that lies between what they would if there was no error and the action taken if no information was revealed. Under uncertainty about the precise parameters of the noisy channel, this underreaction can further be classified. Indeed, the receiver will react less than originally expected to the arrival of new information. The reason for this is a change in their worst-case belief when asked to take action. As it turns out the receiver will put more weight on the belief prior to communication, afraid of having underestimated the probability of the message received due to error in the first place. We give an example illustrating that the interim optimal action corresponds to an ex-ante optimal action under a higher level of noise. Hence, we identify a loss of confidence at the moment a message is received. We prove existence of efficient ex-ante equilibria and show that these are not babbling, i.e., they reveal information. Further research may find the right notion of efficient interim equilibria and prove their existence, as well as investigating the setting under a class of noisy channels without an additively separable structure.

\appendix\label{Appendix: Inertia}

\setcounter{cor}{0}
\renewcommand{\thecor}{\Alph{section}\arabic{cor}}
\setcounter{lem}{0}
\renewcommand{\thelem}{\Alph{section}\arabic{lem}}
\setcounter{proposition}{0}
\renewcommand{\theproposition}{\Alph{section}\arabic{proposition}}
\setcounter{definition}{0}
\renewcommand{\thedefinition}{\Alph{section}\arabic{definition}}

\section{Relation between the error channels from \texorpdfstring{\cite{NT}}{Noisy Talk} and \texorpdfstring{\cite{bauch2021effects}}{GB}}\label{Section: Inertia Appendix Error Channel relation}
\noindent%
In the study of noisy communication, the precise form of the used error channel plays a crucial role. \cite{NT} use a noisy channel that is additively separable, drawing an erroneous message from a distribution independent of the error realization and of the message sent.\footnote{In an extension, they relax this assumption by letting the error probability depend on the sent message.} Under a quadratic loss and uniform state space, they investigate strategic aspects and provide a natural generalization to the well-known model by Crawford and Sobel (\cite{CS}). In contrast, \cite{bauch2021effects} investigates a noisy channel that depends on the word sent under a setting of common interest. Words that are close are more easily confounded. In the following, we analyze their relation, characterizing when they agree in a stochastic setting.\\

To this end, let us first analyze the metric-dependent error channel used in \cite{bauch2021effects}. Let $\mathcal{A}$ be a set of $m+1 \colon = \# \mathcal{A }\geq 2$ elements and consider messages of fixed length $n$, summarized by $\W \colon = \mathcal{A}^n$ with generic elements $\v,\w$. The \emph{Hamming-distance} $d(\v,\w)$ between to words is the number of entries in which they differ. We say that a noisy channel is of \emph{Shannon-type} if there is an error probability $q \in [0,1]$ such that the transition probabilities satisfy
\begin{equation}\label{eq: q-are symmetric channel appendix}
    \epsilonS(\w \mid \v) \colon = (1-q)^{n-d(\v,\w)} \cdot \left(\frac{q}{m}\right)^{d(\w,\w')}.
\end{equation}
The closer the messages are in their Hamming-distance, the higher the probability that the error confounds them. In the literature on error correction this channel is called the $\#\mathcal{A}$-ary symmetric channel of length $n$ with error probability $q$, \cite{Roth}.\\%

We now formalize an error generating process in which the error message is drawn independently of the sent word.  We say that a noisy channel $\epsilonNT$ is of the \emph{noisy-talk-type}, if there are an error distribution $G \in \Delta(\W)$ and for each word $\v \in \W$ an error probability $p_{\v} \in [0,1]$ such that
\begin{equation} \label{eq:NT channel}
\epsilonNT(\w \mid\v) = (1-p_{\v}) \cdot \mathds{1}_{\v}(\w) + p_{\v} \cdot G(\w),
\end{equation}
which is the message-dependent extension of the error channel studied by \cite{NT} and adjusted to the special setting of a finite message space. The intuition behind this noise channel is to have for each word a (possibly distinct) probability $p_{\v}$ with which it is affected by an error. If there is an error, the word received is drawn from a joint error probability distribution $G$.\\%

Obviously, if $p_{\v} = q = 0$ for all $\v$, both channels agree and correctly transmit every message. Likewise, if $p_{\v} = 1$ for all $\v$ and $q = \tfrac{m}{m+1}$, the message received is independent of the sent message and thus \emph{uninformative}, c.f., \cite{bauch2021effects}. The following proposition summarizes that the two ways of modeling amount to the same error generating process only in degenerated cases.
\begin{proposition}\label{Proposition: Shannon and noisy talk channel comparison}
Let $0<q<\tfrac{m}{m+1}$. Then the Shannon-type channel $\epsilonS$ is of noisy-talk-type if and only if $n = 1$.
\end{proposition}

\begin{remark}\label{Remark: uniqueness of Shannon noisy talk channel}
If $n=1$, the noisy channel of Shannon-type can uniquely be written as a noisy-talk-type if and only if $\#\mathcal{A} \geq 3$.
\end{remark}

\section{Proofs and Calculations}\label{Appendix: Inertia: Proofs}
\noindent%
The following purely topological lemma is useful for proving existence of the best replies.
\begin{lem} \label{Lemma: topology supremum continuous if Y is compact}
    Let $X,Y$ be two topological spaces, $Y$ compact and $f \colon X \times Y \to \R$ continuous\footnote{$X \times Y$ endowed with the product topology. Note that this is equivalent to the box topology as it is a finite Cartesian product.}. Then, $g(x) \colon = \sup_y f(x,y)$ is continuous.
    \begin{proof}
        The presented proof is based on the WordPress blog by \cite{williewong}.\\
        We show that $g^{-1}((b,\infty))$ and $g^{-1}((-\infty, a))$ are open in $X$ for all $a,b \in \R$. As the topology of $\R$ is generated by these intervals, this proves continuity of $g$.\smallskip

        First consider $g^{-1}((b,\infty))$. Let $\pi_X \colon X \times Y \to X$ denote the projection which is an open map. Then $g^{-1}((b,\infty)) = \pi_X(f^{-1}((b,\infty)))$ and thus it is open.\\
        Indeed: We have $x \in g^{-1}((b,\infty))$ if and only if there exists a $y \in Y$ with $f(x,y) > b$, i.e., $(x,y) \in f^{-1}((b,\infty))$. This is equivalent to $x \in \pi_X(f^{-1}((b,\infty)))$.\bigskip

        Now consider $g^{-1}((-\infty,a))$. We observe the following equivalence for an $x \in X$: $x \in g^{-1}((-\infty,a))$ if and only if $f(x,y) < a $ for all $y \in Y$ (use compactness of $Y$ to see that $f(x,y)<a \text{ for all } y$ implies $g(x)=\sup_y f(x,y) < a$). This can be written as $(x,y) \in f^{-1}((-\infty,a))$ for all $y$.\\
        In the following we construct for an arbitrary $x \in g^{-1}((-\infty,a))$ an open neighborhood $U_x$ contained in $g^{-1}((-\infty, a))$, thus proving that it is an open set.
        From above, we know that $(x,y) \in f^{-1}((-\infty,a))$ for all $y$. As $f^{-1}((-\infty,a))$ is open in $X \times Y$, we can find open 'box' neighborhoods $x \in U_{(x,y)}\subseteq X$ and $y \in V_{(x,y)} \subseteq Y$ such that $U_{(x,y)} \times V_{(x,y)} \subseteq f^{-1}((-\infty,a))$. Now, $\cup_{y \in Y} V_{(x,y)}$ is an open covering of the compact space $Y$ and thus admits a finite sub-covering $Y' \subseteq Y$. Then, $U_x \colon = \cap_{y \in Y'} U_{(x,y)}$ is an open neighborhood of $x$ as the intersection of finitely many open sets containing $x$. Moreover,
        \begin{equation*}
            U_x \times Y = U_x \times \bigcup_{y \in Y'} V_{(x,y)} = \bigcup_{y \in Y'} U_x \times V_{(x,y)} \subseteq \bigcup_{y \in Y'} U_{(x,y)} \times V_{(x,y)} \subseteq f^{-1}((-\infty, a)).
        \end{equation*}
        Hence, by the above equivalence, we see that $U_x \subseteq g^{-1}((-\infty, a))$ which concludes the proof.        
    \end{proof}
\end{lem}

    \begin{proof}[Proof of Proposition \ref{Proposition: Existence and uniqueness of best replies}]
        Existence of solutions to \eqref{eq: ex ante worst-case receiver simplified} and \eqref{eq: interim worst-case receiver simplified} is immediate by continuity of the involved functionals, c.f., Lemma \ref{Lemma: topology supremum continuous if Y is compact}.\\
        Turning towards uniqueness, consider the ex-ante case first, considering the equivalent problem \eqref{eq: ex ante worst-case receiver similar to interim}. Note that if $\mu_0(C(\w))>0$, for every fixed $G$ the mapping
        \begin{equation*}
            h_G(\alpha) \colon = \sum_{\w} (1-\pup) \mu_0(C(\w))\cdot \E_{C(\w)}[\ell(\norm{t-\alpha(\w)})] + \pup G(\w) \cdot \E_{\mu_0}[\ell(\norm{t-\alpha(\w)})]
        \end{equation*}
        is strictly convex since $T$ is a strictly convex space and $\ell$ is strictly increasing and strictly convex. As $\mathcal{G} \subseteq \Delta(\W)$ is closed and thus compact, the supremum $\sup_{G \in \mathcal{G}} h_G(\alpha)$ is attained and consequently also a strictly convex function. Hence, uniqueness in the ex-ante case follows. If $\mu_0(C(\w))=0$ choosing $\alpha(\w) = \apool$ is optimal as only the term $\pup G(\w) \E_{\mu_0}[\ell(\norm{t-\alpha(\w)})]$ contains $\alpha(\w)$. This choice is only unique if $G(\w) >0$ under the worst-case $G$, otherwise any value is optimal.\\
        Turning towards the interim case, note that if $\mu_0(C(\w))>0$ the Bayesian update is well-defined and expression \eqref{eq: interim worst-case receiver simplified} defines a strictly convex function with unique minimizer as above. If $\mu_0(C(\w))=0$, then $\kappa = 1$ for all $G$ with $G(\w) >0$, leading to $\alpha_{\interim}(\w) = \apool$ and otherwise $\mu^{\pi,p,G}_{\w} = \mu_0$ by convention, again resulting in $\apool$ being optimal.
    \end{proof}

    \begin{proof}[Proof of Proposition \ref{Proposition: Existence of efficient ex-ante wc equilibrium}]
    We first identify the efficient worst-case equilibrium as a solution to the joint minimization problem and then argue that it cannot be babbling.

    To start off, recall that both agents face the same ex-ante worst-case expected loss. A solution to the joint minimization problem thus necessarily involves mutual best replies and guarantees efficiency. We are thus going to show that
        \begin{align*}
            \min_{(\pi,\alpha)} \max_{(p,G) \in \mathcal{P} \times \mathcal{G}} (1-p) \cdot \E_{\mu_0}[\ell(\norm{t-\alpha(\pi(t))})] + p \cdot \sum_{\w} G(\w) \cdot \E_{\mu_0}[\ell(\norm{t-\alpha(\w)})]
        \end{align*}
        is attained by a strategy profile $(\pi, \alpha)$. We can internalize the best reply of the sender, given by \eqref{eq: ex-ante sender optimization} and its subsequent discussion, leading towards solving
        \begin{equation*}
            \min_{\alpha} \  (1-\pup) \cdot \int_T \min_{\v} \ell(\norm{t-\alpha(\v)}) \, \mu_0(\mathrm{d}t) + \pup \max_{G \in \mathcal{G}} \sum_{\w} G(\w) \cdot \int_T \ell(\norm{t - \alpha(\w)})\, \mu_0(\mathrm{d}t).
        \end{equation*}
        To show that the above expression can be minimized by a suitable choice of $\alpha$ it is sufficient to show that the functional is continuous. Applying Lebesgue's dominated convergence theorem we see that all the integrals are continuous in $\alpha$ (noting that $\v$ runs through a finite set). Since $\mathcal{G}$ is compact and the function $T^{\# \W} \times \mathcal{G} \to \R, (\alpha, G) \mapsto \sum_{\w} G(\w) \cdot \int_T \ell(\norm{t - \alpha(\w)})\, \mu_0(\mathrm{d}t)$ is continuous, so is its $\max$ over $\mathcal{G}$ by Lemma \ref{Lemma: topology supremum continuous if Y is compact}, proving continuity. Thus, there exists an $\alpha$ minimizing the above functional and consequently also a profile $(\pi, \alpha)$ minimizing the ex-ante worst-case expected loss.\hfill \#\bigskip
        
        We now construct $(\pi,\alpha)$ which has a loss strictly smaller than $L_0$. Since the ex-ante efficient worst-case equilibrium is the joint minimization problem, it cannot be babbling. To this end, choose a hyperplane splitting $T$ into two convex parts $C_1,C_2$ with $C_1 = T\setminus C_2$ being closed, $\apool \notin C_1$ and $0<\mu_0(C_1)<1$. This is always possible since $\mu_0$ is absolutely continuous w.r.t.\ the Lebesgue measure. Let $\v_1,\v_2 \in \W$ be two distinct words and define $\pi$ by $\pi(t) = \v_i$, if $t \in C_i$. Consider interpretation maps of the following form for any $s \in T$:
        \begin{equation}
        \alpha_s(\w) = \begin{cases}
        s &, \w = \v_1,\\
        \apool &, \text{otherwise}.
        \end{cases}
        \end{equation}
        Define $s^*$ to be the unique minimizer of
        \begin{align}
        &\argmin_s \max_{(p,G) \in \mathcal{P} \times \mathcal{G}}\min_s L(\pi,\alpha_s)\nonumber\\
        =&\argmin_s  \ (1-\pup) \mu_0(C_1) \E_{C_1}[\ell(\norm{t-s})] + \pup \overline{G}(\v_1) \E_{\mu_0}[\ell(\norm{t-s})],
        \end{align}
        which is given by
        \begin{equation}
            s^* = \frac{(1-\pup) \mu_0(C_1)}{(1-\pup) \mu_0(C_1) + \pup \overline{G}(\v_1)} \alpha_{C_1} + \frac{\pup \overline{G}(\v_1)}{(1-\pup) \mu_0(C_1) + \pup \overline{G}(\v_1)} \apool .
        \end{equation}
        Since $C_1$ does not contain $\apool$ and $\mu_0(C_1)>0$ we have $\alpha_{C_1} \neq \apool$ and hence also $s^*\neq \apool$. By construction, thus $\max_{p,G} L(\pi,\alpha)<L_0$.
    \end{proof}

\begin{proof}[Proof of Lemma \ref{Lemma: regular words properties}]
    \emph{Sufficiency:} Revisit the interim problem in expression \eqref{eq: interim worst-case receiver simplified}. For any value of $s$, the worst-case $\kappa$, indicating the weight in the convex combination of the expectations about the cell and the state space, is attained by either $\underline{\kappa}$ or $\overline{\kappa}$. Consequently, the $\max$ can be taken over the set $\{ \underline{\kappa} , \overline{\kappa} \}$. Both convex combinations define a strictly convex function in $s$ with unique local minimum at $\beta_{\kappa} \colon = (1-\kappa) \alpha_{C(\w)} + \kappa \apool$ ($\kappa \in \{ \underline{\kappa} , \overline{\kappa} \}$). Hence the optimal value for $s$ must lie between these two, by the argument given in Proposition \ref{Proposition: Optimizer on line segment generalized in proof}. The worst-case is thus certainly invoked by $\overline{\kappa}$ and leads to $\alpha_{\interim}(\w) = \beta_{\overline{\kappa}}$ if plugging in $\beta_{\overline{\kappa}}$ in both expectation operators still favors the one of the cell $C(\w))$. Formally, this is expressed by
    \[
    \begin{aligned}
        &&\E_{\mu_0,C(\w)}[\norm{t-\beta_{\overline{\kappa}}}_2^2] &\leq \E_{\mu_0}[\norm{t-\beta_{\overline{\kappa}}}_2^2]\\
        \iff&& \Tr_{\mu_0,C(\w)} + \norm{\alpha_{C(\w)} - \beta_{\overline{\kappa}}}_2^2 &\leq \Tr_{\mu_0} + \norm{\apool - \beta_{\overline{\kappa}}}_2^2\\
        \iff&& \Tr_{\mu_0,C(\w)} + \overline{\kappa} \cdot \norm{\alpha_{C(\w)} - \apool}_2^2 &\leq \Tr_{\mu_0} + (1-\overline{\kappa})\cdot\norm{\apool - \alpha_{C(\w)}}_2^2\\
        \iff&& (2\overline{\kappa}-1) \cdot \norm{\alpha_{C(\w)}-\apool}_2^2 &\leq \Tr_{\mu_0} - \Tr_{\mu_0,C(\w)}.
    \end{aligned}
    \]
    \emph{Existence:} Assume by means of contradiction that we have for all $\w$
    \[
    \begin{aligned}
       &&     (2\overline{\kappa}-1) \cdot \norm{\alpha_{C(\w)}-\apool}_2^2 &> \Tr_{\mu_0} - \Tr_{\mu_0,C(\w)}\\
        \implies &&\norm{\alpha_{C(\w)} - \apool}_2^2 &> \Tr_{\mu_0} - \Tr_{\mu_0, C(\w)}\\
        \iff&& \E_{\mu_0,C(\w)}[\norm{t-\apool}_2^2] &> \E_{\mu_0}[\norm{t-\alpha_0}_2^2] = L_0\\
        \iff&& \int_{C(\w)}\norm{t-\apool}_2^2 \mu_0(\mathrm{d}t) &> \mu_0(C(\w)) \cdot \E_{\mu_0}[\norm{t-\alpha_0}_2^2].
    \end{aligned}
    \]
    Adding up over all $\w$ yields the contradiction $L_0 > L_0$.
\end{proof}

    \begin{proof}[Proof of Proposition \ref{Proposition: Optimizer on line segment generalized in proof}]
    If $\mu_0(C(\w)) = 0$ we have $\alpha_{\ex}(\w) = \alpha_{\interim}(\w) = \apool$ which trivially lies on any line segment including $\apool$. Fix thus any $\w$ in the support of $\pi$ in the following.\\
        \emph{Ex-ante best reply:} Under quadratic loss, the ex-ante problem \eqref{eq: ex ante worst-case receiver similar to interim} of the receiver can be re-stated as
        \begin{align*}
            \min_{\alpha} \max_{G} \sum_{\w} (1-\pup) \mu_0(C(\w))\cdot \norm{\alpha_{C(\w)}-\alpha(\w)}_2^2 + \pup G(\w) \cdot \norm{\apool-\alpha(\w)}_2^2.
        \end{align*}
        If now $\alpha(\w)$ did not lie on the line segment between $\alpha_{C(\w)}$ and $\apool$, one can reduce the distance to both, $\alpha_{C(\w)}$ and $\apool$ by choosing the closest point to $\alpha(\w)$ on said line segment, reducing the loss in the above functional.\\
        \emph{Interim best reply:} Under quadratic loss, the interim problem \eqref{eq: interim worst-case receiver simplified} of the receiver is given as a convex combination of two functions in the distance of $s$ to $\alpha_{C(\w)}$ resp.\ $\apool$:
        \begin{align*}
            \min_{s} \max_{p,G} \ &(1-\kappa) \cdot \left( \E_{\mu_0, C(\w)}[\norm{t}_2^2] + \norm{\alpha_{C(\w)}-s}_2^2 - \norm{\alpha_{C(\w)}}_2^2 \right)\\
            &+ \kappa \cdot \left( \E_{\mu_0}[\norm{t}_2^2] + \norm{\apool-s}_2^2 - \norm{\apool}_2^2 \right),
        \end{align*}
        where $\kappa = \frac{p \cdot G(\w)}{(1-p) \mu_0(C(\w)) + pG(\w)}$. If now $s$ was an action which does not lie on the line segment between $\alpha_{C(\w)}$ and $\apool$, switching to the closest point $s^*$ from $s$ which does lie on the line segment strictly decreases the distance to both, $\alpha_{C(\w)}$ and $\apool$. Hence, also the worst-case expected loss is strictly decreased, irrespective of a possible change in $\kappa$ by the induced worst-case beliefs.
        
    \end{proof}

    \begin{proof}[Proof of Theorem \ref{Theorem: interim strictly more pooling than ex ante}]
    We first prove the weak inequality. 
    
        By Proposition \ref{Proposition: Optimizer on line segment generalized in proof}, the ex-ante problem \eqref{eq: ex ante worst-case receiver similar to interim} can be equivalently written as finding appropriate weights $\lambda_{\w} \in [0,1]$ to pin down the solution's position on the line segment between $\alpha_{C(\w)}$ and $\apool$. That way, the problem is reduced to a one dimensional one in each word $\w$, i.e.,
        \begin{equation}
            \argmin_{(\lambda_{\w})_{\w}} (1-\pup) \cdot \sum_{\w} \mu_0(C(\w)) \cdot ((1-\lambda_{\w}) \cdot d_{\w})^2 + \pup \max_{G} \sum_{\w} G(\w) \cdot (\lambda_{\w} \cdot d_{\w})^2,\label{eq: ex ante expressed with percentages}
        \end{equation}
        where $d_{\w} \colon = \norm{\alpha_{C(\w)} - \apool}_2$. The induced actions are thus $\lambda_{\w} \cdot \alpha_{C(\w)} + (1-\lambda_{\w}) \cdot \apool$.
        
        Consider any $\w_0$. If $\mu_0(C(\w)) = 0$ or $d_{\w_0} = 0$, we have $\alpha_{\ex}(\w) = \apool = \alpha_{\interim}(\w)$. Assume thus $\mu_0(C(\w_0))>0$ and $d_{\w_0} > 0$. Now fix a profile $(\lambda_{\w})_{\w}$ of weights and define the auxiliary function $f_{\w_0,G}$ in $\lambda \in [0,1]$ for any fixed $G$ via
        \begin{align*}
            f_{\w_0,G}(\lambda) \colon = &(1-\pup) \cdot \left( \mu_0(C(\w_0) \cdot ((1-\lambda) d_{\w_0})^2 + \sum_{\w \neq \w_0} \mu_0(C(\w)) \cdot ((1-\lambda_{\w}) \cdot d_{\w})^2 \right)\\
            &+ \pup \left( G(\w_0) \cdot (\lambda \cdot d_{\w_0})^2 +  \sum_{\w' \neq \w_0} G(\w') \cdot (\lambda_{w'} \cdot d_{\w'})^2 \right).
        \end{align*}
        The function is strictly convex and its derivative is given by
        \begin{equation}
            f_{\w_0,G}'(\lambda) = 2 \cdot d_{\w_0}^2 \cdot \left( ((1-\pup) \cdot \mu_0(C(\w_0)) + \pup G(\w_0)) \cdot \lambda - (1-\pup) \mu_0(C(\w_0))  \right).\label{eq: derivative aux function}
        \end{equation}
        Since $(1-\pup)\mu_0(C(\w_0)) + \pup G(\w_0) >0$, the function $f_{\w_0, G}$ has a unique minimum given by
        \begin{equation}
            \lambda_{\w_0}(G(\w_0)) \colon = \frac{(1-\pup) \mu_0(C(\w_0))}{(1-\pup) \mu_0(C(\w_0)) + \pup G(\w_0)}.\label{eq: local optimum weight problem in one variable}
        \end{equation}
        
        Recall that the interim optimal action for a regular word corresponds to the value $\lambda_{\interim, \w_0} = \lambda_{\w_0}(\overline{G}(\w_0))$. Evidently, $\lambda_{\interim,\w_0} \leq \lambda_{\w_0}(G(\w_0))$ for all $G$. Denote by $(\lambda_{\ex,\w})_{\w}$ the weights induced by $\alpha_{\ex}$ which is the solution to \eqref{eq: ex ante expressed with percentages}. We thus eventually need to show $\lambda_{\interim,\w_0} \leq \lambda_{\ex,\w_0}$.\\
        
        If $\lambda_{\interim,\w_0} = 0$, we are done. Consider now $0 < \lambda_{\interim, \w_0}$. In the following, we show for any $0< \lambda_0 < \lambda_{\interim, \w_0}$ the inequality $\max_{G} f_{\w_0,G}(\lambda_0) > \max_G f_{\w_0,G} (\lambda_{\interim,\w_0})$. This will imply $\lambda_0 \neq \lambda_{\ex, \w_0}$ for all $\lambda_0 < \lambda_{\interim, \w_0}$ and thus $\lambda_{\interim, \w_0} \leq \lambda_{\ex, \w_0}$ to conclude the proof.\\
        In order to show $\max_{G} f_{\w_0,G}(\lambda_0) > \max_G f_{\w_0,G} (\lambda_{\interim,\w_0})$, we do so for all $G$ uniformly. Recall
        \begin{equation}
            f_{\w_0,G}'(\lambda) = 2 \cdot d_{\w_0}^2 \cdot \left( ((1-\pup) \cdot \mu_0(C(\w_0)) + \pup G(\w_0)) \cdot \lambda - (1-\pup) \mu_0(C(\w_0))  \right).
        \end{equation}
        Plugging in any $\lambda = \lambda_{\interim,\w_0} - \delta \geq 0$ for $\delta >0$ yields
        \begin{equation}
             f'_{\w_0,G}(\lambda) \leq f'_{\w_0,\overline{G}}(\lambda) = - \underbrace{2d_{\w_0}^2 \cdot ((1-\pup) \cdot \mu_0(C(\w_0)) + \pup \overline{G}(\w_0))}_{=\colon C >0} \cdot \delta.
        \end{equation}
        Since the right hand side of the above equation is negative and $f_{\w_0,G}'' <0$ we get the subsequent inequality: By using that the derivative on $[\lambda_0, \tfrac{1}{2} \cdot (\lambda_{\interim,\w}- \lambda_0)]$ is bounded above by $-C \cdot \tfrac{1}{2} \cdot (\lambda_{\interim,\w}- \lambda_0)$ on the interval of length $\tfrac{1}{2} \cdot (\lambda_{\interim,\w}- \lambda_0)$ we obtain
        \begin{equation*}
            f_{\w_0,G}(\lambda_0) \geq f_{\w_0,G}(\lambda_{\interim,\w_0}) + \underbrace{C \cdot \left( \frac{\lambda_{\interim, \w_0} - \lambda_0}{2} \right)^2}_{=\colon C'}.
        \end{equation*}
        Since $C'>0$ is independent of $G$, we find $\max_{G} f_{\w_0,G}(\lambda_0) > \max_G f_{w_0,G} (\lambda_{\interim,\w_0})$.\hfill\#\bigskip

        Let us now address the question of strict inequalities. We already know that $\lambda_{\interim, \w} \leq \lambda_{\ex, \w}$, thus it suffices to prove the existence of a vector $(\lambda_{\w})_{\w}$ that improves upon $\lambda_{\interim, \w}$ in \eqref{eq: ex ante expressed with percentages}. To this end, we compare $(\lambda_{\interim, \w})_{\w}$ to $(c \cdot \lambda_{\interim, \w})_{\w}$ for $c >0$. Note that the $\argmax$-set
        \begin{equation*}
            \argmax_G \sum_{\w} G(\w) \cdot (c \cdot \lambda_{\interim, \w} \cdot d_{\w})^2
        \end{equation*}
        is the same for all $c >0$, especially also for $c = 1$. Thus, restricting our variables to the locus $(c \cdot \lambda_{\interim, \w})_{\w}$ allows us to fix any $G \in \argmax_{G'} \sum_{\w} G'(\w) \cdot ( \lambda_{\interim, \w} \cdot d_{\w})^2$ and we have that the following expression
        \begin{align}
            &(1-\pup) \cdot \sum_{\w} \mu_0(C(\w)) \cdot ((1-c \cdot \lambda_{\interim, \w}) \cdot d_{\w})^2 + \pup \max_{G'} \sum_{\w} G'(\w) \cdot (c \cdot \lambda_{\interim, \w} \cdot d_{\w})^2\nonumber\\
            = &(1-\pup) \cdot \sum_{\w} \mu_0(C(\w)) \cdot ((1-c \cdot \lambda_{\interim, \w}) \cdot d_{\w})^2 + \pup \sum_{\w} G(\w) \cdot (c \cdot \lambda_{\interim, \w} \cdot d_{\w})^2\label{eq: ex ante problem kappa}
        \end{align}
        is smooth in $c>0$. If $(\lambda_{\interim, \w})_{\w}$ was optimal in the ex-ante problem \eqref{eq: ex ante expressed with percentages}, $c=1$ would be a local minimum of \eqref{eq: ex ante problem kappa}. Checking the first-order-condition, however, yields the unique local minimum at
        \begin{equation*}
            c^* = \frac{\sum_{\w} \lambda_{\interim,\w} \cdot d_{\w}^2\cdot (1-\pup) \cdot \mu_0(C(\w))}{\sum_{\w} \lambda_{\interim,\w}^2 \cdot d_{\w}^2 \cdot ((1-\pup)\cdot \mu_0(C(\w)) + \pup \cdot G(\w))}.
        \end{equation*}
        Note that for each summand in the denominator, we have the inequality
        \begin{align*}
            &\lambda_{\interim,\w}^2 \cdot d_{\w}^2 \cdot ((1-\pup)\cdot \mu_0(C(\w)) + \pup \cdot G(\w))\\
            \leq& \lambda_{\interim,\w} \cdot \lambda_{\interim,\w} \cdot d_{\w}^2 \cdot ((1-\pup)\cdot \mu_0(C(\w)) + \pup \cdot \overline{G}(\w))\\
            =& \lambda_{\interim,\w} \cdot d_{\w}^2 \cdot (1-\pup)\cdot \mu_0(C(\w)),
        \end{align*}
        where the inequality must be strict for some $\w$ since $\# \mathcal{G}\geq 2$. Consequently, $c^* > 1$, which shows that $(\lambda_{\interim, \w})_{\w}$ is not optimal for the ex-ante problem.
    \end{proof}

\begin{proof}[Calculations to Example \ref{Example: Inconsistency on unit interval}]
    To find the Pareto-efficient ex-ante equilibrium, we proceed as in the proof of Proposition \ref{Proposition: Existence of efficient ex-ante wc equilibrium}. We make the action vector of receiver a variable and briefly write $a = \alpha(L), b = \alpha(R) \in T$ and may assume $a \leq b$ (otherwise exchange). Any $(a,b)$ defines a best response of the sender, cutting the interval in half at the midpoint $m = \tfrac{a+b}{2}$. To solve the ex-ante problem we may thus treat
    \begin{align}
        \min_{(a,b)} \ &(1-p) \left( \int_{-\tfrac{1}{2}}^m (t-a)^2 \, \mathrm{d}t +  \int_{m}^{\tfrac{1}{2}} (t-b)^2 \, \mathrm{d}t \right)\nonumber\\
        + &p \max_{g \in [0,1]} \left\{ (1-g) \cdot \int_{\tfrac{1}{2}}^{-\tfrac{1}{2}} (t-a)^2 \, \mathrm{d}t, g \cdot \int_{\tfrac{1}{2}}^{-\tfrac{1}{2}} (t-b)^2 \, \mathrm{d}t \right\}.
    \end{align}
    Using $\int_{\tfrac{1}{2}}^{-\tfrac{1}{2}} (t-s)^2 = s^2 + \Tr_{\mu_0,[-\tfrac{1}{2},\tfrac{1}{2}]}$, we can simplify the problem to
    \begin{equation}
        \min_{(a,b)} \ (1-p) \left( \int_{-\tfrac{1}{2}}^m (t-a)^2 \, \mathrm{d}t +  \int_{m}^{\tfrac{1}{2}} (t-b)^2 \, \mathrm{d}t \right) + p \max \{ a^2,b^2\}.\label{eq: example simplified interval}
    \end{equation}
    To better handle the $\max$-function, we distinguish two cases. Either $a^2 = b^2$ or not. In the following, we argue that $a^2 \neq b^2$ cannot lead to a Pareto-efficient equilibrium.\\
    To this end, we may assume $\betrag{a} < \betrag{b}$ (if there is an equilibrium with the inverse strict inequality, there is also one considering $s \mapsto -s$ by symmetry), i.e., $\max\{ a^2,b^2 \} = b^2$. In equilibrium, $a$ should be a minimum in \eqref{eq: example simplified interval}. Taking first-order conditions, we find the only local minimum at $a = \tfrac{1}{3}(b-1)$ (the other is a maximum). Especially, $a\leq0\leq b$. We plug this back into the problem to get a function in $b$ only, where we restrict to $b \in [\tfrac{1}{4},\tfrac{1}{2}]$ to maintain $\betrag{a} < \betrag{b}$. Now, we find two local extreme points, the only minimum being at $b^* = -\tfrac{4 \, p - 3 \, \sqrt{8 \, p + 1} + 5}{8 \, {\left(p - 1\right)}}$. However, the respective $a^* = 1/3 (b^*-1)$ is larger in absolute value for all $p \in [0,1]$. Consequently, there is no equilibrium with $a^2 \neq b^2$.\\
    Consider now $a^2 = b^2$. If $a=b$, they must be equal to $\apool = 0$ leading to a babbling equilibrium. Assume thus $a=-b$, which always leads to the communication device $\pi$ stated. From \eqref{eq: example simplified interval} taking first-order conditions leads to $b^* = \tfrac{1}{4}(1-p) = -a^*$. Indeed, this equilibrium outperforms the babbling equilibrium and must thus be the ex-ante Pareto-efficient equilibrium.\hfill \#\\

    Fix $\pi$ and turn towards the interim optimal actions. We apply the sufficient criterion decoded in Lemma \ref{Lemma: regular words properties}. Indeed, we have for, say, $\w = R$
    \begin{align}
        \Tr_{\mu_0} - \Tr_{\mu_0, [0,\tfrac{1}{2}]} &= \tfrac{1}{12} (1 - \tfrac{1}{4})= \tfrac{1}{16} \geq (2\overline{\kappa} - 1) \tfrac{1}{16} = (2\overline{\kappa} - 1) (\alpha_{C(R)} - \apool)^2.
    \end{align}
    Note that $\overline{\kappa} = \tfrac{p \cdot 1}{p \cdot 1 + (1-p) \tfrac{1}{2}} = \frac{2p}{1+p}$. Consequently, we have $\alpha_{\interim}(R) = (1-\overline{\kappa}) \alpha_{C(R)} + \overline{\kappa} \apool = \tfrac{1}{4} \tfrac{1-p}{1+p}$, similarly $\alpha_{\interim}(L) = - \alpha_{\interim}(R)$. Especially, the action profile and $\pi$ together form an interim equilibrium.
\end{proof}

\begin{proof}[Proof of Proposition \ref{Proposition: Shannon and noisy talk channel comparison}]
\emph{Sufficiency:} Given $n=1$, define $\epsilonNT$ by setting $G(\w) \colon = (m+1)^{-1}$ and $p_{\v}=p\colon = q \cdot \tfrac{m+1}{m} \in (0,1)$ and check for $\w,\v \in \W, \w \neq \v$
\begin{align}
&\epsilonNT(\v \mid \v) = 1-p + \frac{p}{m+1} = 1- q \cdot \frac{m+1}{m} + \frac{q}{m} = 1-q = \epsilonS(\v \mid \v),&\\
&\epsilonNT(\w \mid \v) = \frac{p}{m+1} = \frac{q}{m} = \epsilonS(\w \mid \v).&
\end{align}
Thus $\epsilonS(\w \mid \v) = \epsilonNT(\w \mid \v)$ for all $\w,\v \in \W$.\bigskip

\emph{Necessity:} Assume $n \geq 2$ and assume by means of contradiction that there are $q,(p_{\v})_{\v},G$ such that $\epsilonS \equiv \epsilonNT$. Note, that we cannot have $p_{\v} = 0$ for any $\v$. Otherwise we have $\mathds{1}_{\v}(\w) = \epsilonNT(\w \mid \v) = \epsilonS(\w \mid \v) = (1-q)^{n-d(\w,\v)} \cdot \left( \tfrac{q}{m+1} \right)^{d(\w,\v)}$. Plugging in $\w = \v$ and any $\w \neq \v$ we convince ourselves that this is only possible if $q = 0$ which was excluded. Likewise, we cannot have $G(\v) = 0$ for any $\v$. Otherwise $\epsilonNT(\w \mid \v) = (1-p_{\w}) \mathds{1}_{\w}(\v)$. Plugging in any $\w \neq \v$ this is zero, which is only possible for the Shannon-type channel if $q=0$ which we excluded. We can thus subsequently assume $p_{\v}, G(\v) \neq 0$ for all $\v$.\\%
We now make the following observation: Let $\v,\w,\w'$ be words with $d(\w,\v) = d(\w',\v)$. Then $p_{\v} G(\w) = \epsilonNT(\w \mid\v) = \epsilonS(\w \mid \v) = \epsilonS(\w'\mid\v)=\epsilonNT(\w'\mid\v) = p_{\v}G(\w')$ and thus $G(\w) = G(\w')$. By using $\epsilonNT(\v \mid\w) = \epsilonNT(\v \mid \w')$ we similarly obtain $p_{\w} = p_{\w'}$.\\
Assume for the moment that $\#\mathcal{A} \geq 3$ and take any $\w,\w'$. Then there always is a $\v$ with $d(\w,\v) = d(\w',\v)$. Indeed, one can pick $\v$ to be the word the entries $\v_i$ of which are equal to $\w_i$ if $\w_i = \w'_i$ and choose a letter different from both $\w_i$ and $\w'_i$ otherwise. By the argument above we find $p \colon = p_{\w} = p_{\w'}$ and $g \colon = G(\w) = G(\w')$. Since $\w,\w'$ were arbitrary, they must be the same for all words. However, now taking pairwise different $\v,\w,\w'$ with $d(\w,\v) = 1, d(\w,\v) = 2$, which is possible since $n \geq 2$, we find $\epsilonNT(\w \mid\v) =pg = \epsilonNT(\w'\mid \v)$, but $\epsilonS(\w \mid \v) = (1-q)^{n} \cdot \widetilde{q} \neq (1-q)^{n} \cdot \widetilde{q}^{2} = \epsilonS(\w' \mid \v)$ where $\widetilde{q} \colon = \tfrac{q}{(1-q)m}>0$.\\
Assume now $\mathcal{A} = \{ 0,1\}$. We define an equivalence relation on $\W$ via $\v \sim \w$ if $2 \mid d(\w,\v)$. Indeed, observe that $\sum_i \v_i - \sum_i \w_i \equiv_2 \sum_{i, \v_i \neq \w_i} 1  = d(\w, \v)$ to show transitivity, where $\equiv_2$ means $\text{mod} \, 2$. The message space thus splits into two equivalence classes, represented by $\w_0 = (0,0,\dots)$ and $\w_1 = (1,0,\dots)$. We have that $g_0 \colon = G(\w_0)$, $p_0 \colon = p_{\w_0}$ are the same for all elements of $\w_0$'s equivalence class, similarly for $\w_1$. Each class contains $2^{n-1}$ different words and since $G$ is a distribution, we find $g_0+g_1 = 2^{1-n}$. Furthermore $p_1 g_0 = \epsilonNT(\w_0 \mid \w_1) = \epsilonS(\w_0 \mid \w_1) = \epsilonS(\w_1 \mid \w_0) = \epsilonNT(\w_1 \mid \w_0) = p_0 g_1$, giving us $\tfrac{p_0}{p_1} = \tfrac{g_0}{g_1}$ and $1-p_0 + p_0 g_0 = \epsilonNT(\w_0 \mid \w_0) = \epsilonS(\w_0 \mid \w_0) = \epsilonS(\w_1 \mid \w_1) = \epsilonNT(\w_1 \mid \w_1) = 1-p_1 + p_1 g_1$ yielding $\tfrac{p_0}{p_1} = \tfrac{1-g_1}{1-g_0}$. From $\tfrac{g_0}{g_1} = \tfrac{1-g_1}{1-g_0}$ we obtain $(g_0-g_1)(g_0+g_1) = g_0-g_1$ and thus $g_0 = g_1$ since $g_0+g_1 \neq 0$. Hence also $p_0 = p_1$ and we find the same contradiction as in the previous argument.
\end{proof}

\begin{proof}[Proof of Remark \ref{Remark: uniqueness of Shannon noisy talk channel}]
Recall the proof of Proposition \ref{Proposition: Shannon and noisy talk channel comparison} and observe that for distinct $a,b \in \mathcal{A}$ we have $q = m \cdot p_a \cdot G(b)$ and $(1-q) = 1- p_a + p_a \cdot G(a)$. Especially, $p_a>0$. Combining this and both the equations we find $1 = G(a) + m\cdot G(b)$ for all $a,b \in \W, a\neq b$. Thus $G(a) + m \cdot G(b) = G(b) + m \cdot G(a)$ which implies $G(a) = G(b)$ if $m \geq 2$. But then also $p_a=p_b$ and the choice from the sufficiency part is the only one possible. If $n=1$ and $m=1$ we can choose any $G(a) \in (0,1)$ and get solutions for $G(b) \colon = 1-G(a), p_b = \tfrac{q}{G(a)}, p_a = \tfrac{q}{G(b)}$.
\end{proof}

\bibliography{Literatur}{}

\end{document}